\documentclass[a4paper,journal]{IEEEtran}

\usepackage[dvipsnames]{xcolor}

\IEEEoverridecommandlockouts
\usepackage[utf8]{inputenc}
\usepackage[cmex10]{amsmath}
\DeclareMathOperator{\diag}{diag}
\usepackage{cleveref}
\usepackage{amsthm}
\usepackage{subfigure}
\usepackage{booktabs}
\usepackage{amsfonts}
\usepackage{amssymb}
\usepackage{bbm}
\usepackage{tikz}
\usepackage{physics}
\usepackage{multirow}

\usepackage[short]{optidef}
\usepackage{mathtools}
\usepackage[noadjust]{cite}
\usepackage{textcomp}
\usetikzlibrary{shapes,arrows}
\usepackage{comment}
\usepackage{algorithm}
\usepackage{algpseudocode}

\usepackage{float}
\floatstyle{plaintop}

\usepackage{caption}
\usepackage{tikz}

\usepackage{colortbl}

\usepackage{multicol}

\newtheorem*{theorem*}{Theorem}

\newtheorem{proposition}{Proposition}

\title{\Large Optimization of RIS-Aided MIMO -- A Mutually Coupled Loaded Wire Dipole Model}
\author{H.~El~Hassani, X.~Qian, S. Jeong, N. S. Perovi\'c, M. Di Renzo, \IEEEmembership{Fellow,~IEEE}, P. Mursia, \IEEEmembership{Member,~IEEE}, V. Sciancalepore, \IEEEmembership{Senior~Member,~IEEE}, and X. Costa-Pérez, \IEEEmembership{Senior~Member,~IEEE} \vspace{-0.90cm}
\thanks{Manuscript received June 15, 2023; revised Sep. 17, 2023. S. Jeong, N. S. Perovi\'c, M. Di Renzo are with Universit\'e Paris-Saclay, CNRS, CentraleSup\'elec, Laboratoire des Signaux et Syst\`emes, 91192 Gif-sur-Yvette, France. (marco.di-renzo@universite-paris-saclay.fr). H. El Hassani and X. Qian were with Universit\'e Paris-Saclay, CNRS, CentraleSup\'elec, when this work was done. P. Mursia, V. Sciancalepore are with NEC Labs Europe, Germany. X. Costa-Pérez is with the i2cat Research Center, Spain, and with NEC Labs Europe, Germany. This work was supported in part by the European Commission through the projects H2020 ARIADNE (grant 871464), H2020 RISE-6G (grant 101017011), H2020 SURFER (grant 101030536), the NEC Fellowship Program, and the French ANR through the PEPR-5G project.
}
}

\begin{document}

\maketitle

\begin{abstract}
We consider a reconfigurable intelligent surface (RIS) assisted multiple-input multiple-output (MIMO) system in the presence of scattering objects. The MIMO transmitter and receiver, the RIS, and the scattering objects are modeled as mutually coupled thin wires connected to load impedances. We introduce a novel numerical algorithm for optimizing the tunable loads connected to the RIS, which does not utilize the Neumann series approximation. The algorithm is provably convergent, has polynomial complexity with the number of RIS elements, and outperforms the most relevant benchmark algorithms while requiring fewer iterations and converging in a shorter time.
\end{abstract}
\vspace{-0.20cm}

\vspace{-0.40cm}
\section{Introduction}
\IEEEPARstart{T}{he} reconfigurable intelligent surface (RIS) is an emerging technology that has the ability of smartly controlling the propagation environment without the need of power amplifiers or digital signal processing units \cite{MDR_JSAC}. In RIS-aided communications, it is essential to utilize scattering models that are electromagnetically consistent \cite{MDR_PIEEE} and hence account for the mutual coupling caused by the sub-wavelength design \cite{marco_EuCAP}.

To this end, a communication model for RIS-aided systems that explicitly considers the mutual coupling was introduced in \cite{gradoni2021end}. The model relies on a loaded thin wire dipole approximation for the reconfigurable elements of the RIS \cite{DDA}. Departing from \cite{gradoni2021end}, the authors of \cite{xuewen_SISO} proposed a framework for optimizing the tunable impedances of an RIS-aided single-input single-output (SISO) system. Then, the authors of \cite{abrardo_MIMO} generalized the algorithm in \cite{xuewen_SISO} for application to multi-user and multi-RIS multiple-input multiple-output (MIMO) systems. Recently, the authors of \cite{placido_MISO} generalized the model in \cite{gradoni2021end} by including, in an electromagnetically consistent manner, the impact of scattering objects in the environment. Also, an algorithm that optimizes the sum-rate was introduced.

The algorithms in \cite{xuewen_SISO,abrardo_MIMO,placido_MISO} rely on Neumann's series approximation for optimizing the impedances of the RIS. The novelty and contribution of this paper are to develop an algorithm for RIS-aided MIMO systems that does not utilize any approximation. The proposed approach optimizes the tunable load impedances one by one and iteratively. At each iteration, a closed-form solution is provided by applying Gram-Schmidt's orthogonalization method. Using numerical simulations, we compare the proposed approach with those in \cite{abrardo_MIMO,placido_MISO}, and show that it requires fewer iterations and less time to converge.

\textit{Notation}:
Matrices/vectors are denoted by capital/small bold fonts. $j$ is the imaginary unit. $|z|$, $z^*$, $\Re{z}$, $\Im{z}$ denote the absolute value, conjugate, real, imaginary parts of $z$. $\mathbb{E}\{\cdot\}$ denotes the expectation. $\lVert \cdot \rVert$ denotes the $l_2$-norm. $\diag(\mathbf{a})$ is the square diagonal matrix with the elements of $\mathbf{a}$ on the main diagonal. $\mathbf{A}^{-1}$, $\mathbf{A}^{T}$, $\mathbf{A}^{H}$, $\det(\mathbf{A})$, $\tr(\mathbf{A})$, $\rank(\mathbf{A})$ denote the inverse, transpose, hermitian, determinant, trace, rank of $\mathbf{A}$. $\mathbf{A}(i,j)$ is the $(i,j)$th element of $\mathbf{A}$. $\langle \mathbf{a},\mathbf{b}\rangle$ is the inner product between $\mathbf{a}$ and $\mathbf{b}$. $\mathbf{I}_N$ and $\mathbf{0}_N$ denote the identity and zero matrices of size $N \times N$. $\mathcal{O}(\cdot)$ stands for the big-O notation.

\vspace{-0.25cm}
\section{System Model}
{\textcolor{black}{We consider an RIS-aided MIMO system that comprises a transmitter with $M$ antennas, a receiver with $L$ antennas, and an RIS with $N_{\rm RIS}$ tunable scattering elements, similar to \cite[Fig. 1]{placido_MISO}}}. Based on \cite{gradoni2021end}, the transmit and receive antennas, and the scattering elements of the RIS are modeled as cylindrical thin wire dipoles of perfectly conducting
material whose length is $l$ and whose radius $a \ll l$ is finite but negligible. The thin dipoles of the RIS are controlled by tunable complex-valued impedances. The inter-distance between adjacent scattering elements is denoted by $d \le \lambda/2$, where $\lambda$ is the wavelength. Similar to \cite{placido_MISO}, we consider $N_{e}$ scattering objects that mimic multipath propagation. They are distributed in clusters and are modeled as cylindrical thin wire dipoles of perfectly conducting material and are connected to specified  impedances. The length of the wires and the impedances depend on the electromagnetic properties of the material objects. If the scattering object is a metal plate, e.g., the impedances are equal to zero. {\textcolor{black}{Compared with the statistical multipath channel utilized in \cite{abrardo_MIMO}, the multipath model introduced in \cite{placido_MISO} is electromagntically consistent. Notably, the impact of the scattering objects is not a pure additive term as in \cite{abrardo_MIMO}}}.

Based on \cite{placido_MISO}, the end-to-end channel can be expressed as{\textcolor{black}{{\footnote{{\textcolor{black}{As detailed in \cite{MDR_PIEEE}, \eqref{channel} is valid in the near field of the RIS. Thus, it provides consistent results as the number of RIS elements increases without bound.}}}}}}
\begin{equation}
\mathbf{H}_{\rm E2E} =\mathbf{Z}_{\rm RL} \left[ {\mathbf{Z}_{\rm ROT} - \mathbf{Z}_{\rm ROS} \mathbf{Z}_{\rm sca} \mathbf{Z}_{\rm SOT}} \right] \mathbf{Z}_{\rm TG} \label{channel}
\end{equation}
where $\mathbf{Z}_{\rm RL}= \left(\mathbf{I}_L+\mathbf{Z}_{\rm RR} \mathbf{Z}_L^{-1}\right)^{-1}$, $\mathbf{Z}_{\rm TG}=\left(\mathbf{Z}_{\rm TT}+\mathbf{Z}_G\right)^{-1}$, and $\mathbf{Z}_{\rm sca}= \left(\mathbf{Z}_{\rm SS} + \mathbf{Z}_{\rm SOS} +\mathbf{Z}_{\rm RIS} \right)^{-1}$. Specifically, $\mathbf{Z}_G \in \mathbb{C}^{M \times M}$ and $\mathbf{Z}_L \in \mathbb{C}^{L \times L}$ are diagonal matrices containing the impedances of the voltage generators at the transmitter and the load impedances at the receiver; $\mathbf{Z}_{\rm TT}\in \mathbb{C}^{M \times M}$ and $\mathbf{Z}_{\rm RR}\in \mathbb{C}^{L \times L}$ are the matrices containing the self and mutual impedances at the transmitter and receiver; $\mathbf{Z}_{\rm ROT} \in \mathbb{C}^{L \times M}$, $\mathbf{Z}_{\rm ROS} \in \mathbb{C}^{L \times N_{\rm RIS}}$, $\mathbf{Z}_{\rm SOS} \in \mathbb{C}^{N_{\rm RIS} \times N_{\rm RIS}}$ and $\mathbf{Z}_{\rm SOT} \in \mathbb{C}^{N_{\rm RIS} \times M}$ are the matrices containing the mutual impedances between  different network elements, with $\{\rm{T},R,S,O\}$ denoting the transmitter, receiver, RIS and scattering objects.

From \cite[Eqs. (9)-(13)]{placido_MISO}, we obtain $\mathbf{Z}_{\rm ROT}= \mathbf{Z}_{\rm RT}-\mathbf{Z}_{\rm RO}\mathbf{\Bar{Z}}^{-1}_{\rm OO}\mathbf{Z}_{\rm OT}$, $\mathbf{Z}_{\rm ROS}= \mathbf{Z}_{\rm RO} \mathbf{\Bar{Z}}^{-1}_{\rm OO} \mathbf{Z}_{\rm OS} - \mathbf{Z}_{\rm RS}$, $\mathbf{Z}_{\rm SOS}= - \mathbf{Z}_{\rm SO}\mathbf{\Bar{Z}}^{-1}_{\rm OO} \mathbf{Z}_{\rm OS}$, and $\mathbf{Z}_{\rm SOT}= \mathbf{Z}_{\rm SO}\mathbf{\Bar{Z}}^{-1}_{\rm OO}\mathbf{Z}_{\rm OT}-\mathbf{Z}_{\rm ST}$, where $\mathbf{\Bar{Z}}_{\rm OO}=\mathbf{Z}_{\rm OO}+\mathbf{Z}_{\rm US}  \in \mathbb{C}^{N_e \times N_e }$, with $\mathbf{Z}_{\rm US}$ being the diagonal matrix containing the material-dependent load impedances of the scattering objects, which are assumed given and fixed.

The remaining matrices are the self and mutual impedances between pairs of network elements. Also, $\mathbf{Z}_{\rm SS} \in \mathbb{C}^{N_{\rm RIS} \times N_{\rm RIS}}$ is the matrix of self and mutual impedances between pairs of RIS elements. When the mutual coupling is negligible, $\mathbf{Z}_{\rm SS}$ is a (dominant) diagonal matrix. {\textcolor{black}{The impedance matrices in \eqref{channel} can be computed by using either the framework in \cite{marco_EuCAP} or full-wave simulations. The analytical solution in \cite{marco_EuCAP} is typically preferable, especially if $\mathbf{Z}_{\rm SS}$ is an optimization variable \cite{ZelayaMH21}}}. Finally, $\mathbf{Z}_{\rm RIS} \in \mathbb{C}^{N_{\rm RIS} \times N_{\rm RIS}}$ is a diagonal matrix whose entries are the tunable impedances of the RIS elements, which are to be optimized. $\mathbf{Z}_{\rm RIS}$ can be expressed as $\mathbf{Z}_{\rm RIS}= \diag \left(\{R_{0,k}+j X_{k}\}_{k=1}^{N_{\rm RIS}}\right)$, where $R_{0,k} \geq 0$ is the parasitic resistance that models the internal losses of the $k$th RIS element, which is assumed fixed, and $X_{k} \in \mathcal{P}$ is the reactance of the $k$th load impedance, whose value lies in the feasible set $\mathcal{P}=[ X_{\ell b},X_{ub}] \subset \mathbb{R}$ and is to be optimized.

Let $ \mathbf{x} \in \mathbb{C}^{M \times 1}$ be the transmitted vector. The transmit covariance matrix is $\mathbf{Q}=\mathbb{E}\{\mathbf{x}\mathbf{x}^H\} \in \mathbb{C}^{M \times M}$, with $\mathbf{Q}$ a positive semi-definite matrix, i.e., $\mathbf{Q} \succcurlyeq 0$. We consider the average sum power constraint $\tr(\mathbf{Q}) \leq P_t$, where $P_t$ is the total power. Thus, the received vector is $\mathbf{y}=\mathbf{\mathbf{H}}_{\rm E2E} \mathbf{x}+\mathbf{n}$, where $\mathbf{n} \sim \mathcal{CN}(0,\sigma^2 \mathbf{I}_L)$ is the circularly symmetric complex Gaussian noise vector with zero mean and variance $\sigma^2$.

Therefore, the achievable data rate can be formulated as
\begin{equation}
    R (\mathbf{Q},\mathbf{Z}_{\rm RIS})= \log_2 \left[ \det(\mathbf{I}_L+\frac{\mathbf{H}_{\rm E2E} \mathbf{Q} \mathbf{H}_{\rm E2E}^H }{\sigma^{2}})\right] \label{rate}
\end{equation}

\vspace{-0.25cm}
\section{Problem Formulation and Solution}
We aim to maximize the achievable rate in \eqref{rate} as a function of $\mathbf{Q}$ and $\mathbf{Z}_{\rm RIS}$. Specifically, we have ($k\in\{1,\hdots,N_{\rm RIS}\}$)
\begin{alignat}{3}
\label{opt:P1}
\mathbf{(P0)} \quad \max_{\mathbf{\mathbf{Q},\mathbf{Z}}_{\rm RIS}}& \ R (\mathbf{Q},\mathbf{Z}_{\rm RIS}) \\
\text{s.t. }\quad  &  \Re\{\mathbf{Z}_{\rm RIS}(k,k)\} = R_{0,k} \geq 0, \; \; \forall k
& \\
              & \Im\{\mathbf{Z}_{\rm RIS}(k,k)\} \in \mathcal{P}, \; \; \forall k & \\
               & \tr(\mathbf{Q}) \leq P_t, \quad \mathbf{Q} \succcurlyeq  \mathbf{0} &
\end{alignat}

The formulated optimization problem is non-convex due to the joint optimization of the transmit covariance matrix $\mathbf{Q}$ and the matrix of tunable impedances $\mathbf{Z}_{\rm RIS}$. To tackle it, we introduce an iterative algorithm based on the alternating optimization (AO) method, which decouples $\mathbf{(P0)}$ into two-sub-problems. First, $\mathbf{(P0)}$ is solved with respect to $\mathbf{Q}$ while keeping $\mathbf{Z}_{\rm RIS}$ fixed, and then $\mathbf{(P0)}$ is solved with respect to $\mathbf{Z}_{\rm RIS}$ while keeping $\mathbf{Q}$ fixed. The details are given next.

\vspace{-0.25cm}
\subsection{Optimization of $\mathbf{Q}$}
By keeping $\mathbf{Z}_{\rm RIS}$ fixed, $\mathbf{(P0)}$ boils down to a conventional MIMO optimization problem \cite{waterfill}. Specifically,
let $\mathbf{H}_{\rm E2E}=\mathbf{U}_{\mathbf{H}_{\rm E2E}} \mathbf{\Sigma}_{\mathbf{H}_{\rm E2E}} \mathbf{V}_{\mathbf{H}_{\rm E2E}}^H$ be the singular value decomposition of $\mathbf{H}_{\rm E2E}$, where $\mathbf{V}_{\mathbf{H}_{\rm E2E}} \in \mathbb{C}^{M \times D}$, $\mathbf{U}_{\mathbf{H}_{\rm E2E}} \in \mathbb{C}^{L \times D}$, and $D=\rank(\mathbf{H}_{\rm E2E}) \leq \min(L,M)$. Then, the optimal $\mathbf{Q}^{\star}$ is
\begin{equation}
\mathbf{Q}^{\star}=\mathbf{V}_{\mathbf{H}_{\rm E2E}} \diag(p_1^{\star}, \hdots, p_D^{\star}) \mathbf{V}_{\mathbf{H}_{\rm E2E}}^H \label{waterfilling}
\end{equation}
where $p_i^{\star}=\max \left(\left( 1/\alpha - \sigma^2 \big / \mathbf{\Sigma}_{\mathbf{H}_{\rm E2E}}(i,i)^2 \right),0 \right)$, with $\alpha$ satisfying $\sum_{i=1}^D p_i^{\star} = P_t$ (water-filling power allocation).

\vspace{-0.25cm}
\subsection{Optimization of $\mathbf{Z}_{\rm RIS}$}
By keeping $\mathbf{Q}$ fixed, the resulting optimization problem with respect to $\mathbf{Z}_{\rm RIS}$ simplifies to ($k\in\{1,\hdots,N_{\rm RIS}\}$)
\begin{alignat}{3}
\mathbf{(P1)} \quad \max_{\mathbf{\mathbf{Z}}_{\rm RIS}}& \ \log_2 \left[ \det(\mathbf{I}_L+\frac{\mathbf{H}_{\rm E2E} \mathbf{Q} \mathbf{H}_{\rm E2E}^H }{\sigma^2})\right] \\
\text{s.t. }\quad  &  \Re\{\mathbf{Z}_{\rm RIS}(k,k)\} = R_{0,k} \geq 0 , \; \; \forall k & \\
              & \Im\{\mathbf{Z}_{\rm RIS}(k,k)\} \in \mathcal{P} , \; \; \forall k &
\end{alignat}

In \cite{xuewen_SISO,abrardo_MIMO,placido_MISO}, $\mathbf{(P1)}$ is tackled by capitalizing on the Neumann series approximation, which offers a first-order linear approximation for $\mathbf{Z}_{\rm sca}$ as a function of $\mathbf{Z}_{\rm RIS}$. We circumvent the Neumann series approximation by devising a new approach that combines Sherman-Morrison's inversion formula, Sylvester's determinant theorem, and, more importantly, Gram-Schmidt's orthogonalization method. Specifically, the proposed approach exploits the block coordinate descent (BCD) method \cite[Subsec. 2.7]{convergence1}, which, at the $k$th step, updates the $k$th  tunable impedance $\mathbf{Z}_{\rm RIS} (k,k)$, while keeping all the other impedances fixed and setting them to their most recently updated values.

We depart from $\mathbf{Z}_{\rm sca}$ and decouple the $k$th tunable impedance $\mathbf{Z}_{\rm RIS} (k,k)$ to be optimized from all the other impedances that are kept fixed. Accordingly, we write
\begin{equation}
\mathbf{Z}_{\rm sca}= \left( \mathbf{Z}_{\rm SS} + \mathbf{Z}_{\rm SOS} + \mathbf{Z}_{{\rm RIS},k} +  \mathbf{Z}_{\rm RIS}(k,k)  \mathbf{e}_k \mathbf{e}_k^T \right)^{-1} \label{1steq}
\end{equation}
where $\mathbf{Z}_{{\rm RIS},k}$ denotes the matrix $\mathbf{Z}_{\rm RIS}$ with $\mathbf{Z}_{\rm RIS}(k,k)=0$, and $\mathbf{e}_k$ denotes the vector whose entries are all zeros except the $k$th entry that is set equal to one. We aim to optimize the elements of $\mathbf{Z}_{\rm RIS}$ one by one, i.e., at the $k$th step, we optimize $\mathbf{Z}_{\rm RIS}(k,k)$ and keep all the other elements in $\mathbf{Z}_{{\rm RIS},k}$ fixed.

For ease of presentation, we introduce the notation
\begin{equation}
    \mathbf{A}_k =\mathbf{Z}_{\rm SS}+ \mathbf{Z}_{\rm SOS}+\mathbf{Z}_{{\rm RIS},k}, \quad    z_k=\mathbf{Z}_{\rm RIS}(k,k)
\end{equation}
Also, $\mathbf{A}_k$ and $\mathbf{A}_k + z_k\mathbf{e}_k \mathbf{e}_k^T$ are assumed to be invertible matrices, which is ensured by the physical nature of the problem and can be tested during the execution of the algorithm.

By applying the Sherman-Morrison formula \cite[Subsec. 2.7.1]{book_sciComputing} to \eqref{1steq}, the matrix $\mathbf{Z}_{\rm sca}$ can be written as
\begin{equation}
    \mathbf{Z}_{\rm sca} = \mathbf{Z}_{\rm sca} \left(z_k\right) = \mathbf{A}_k^{-1}- \frac{ \mathbf{A}_k^{-1} \mathbf{e}_k \mathbf{e}_k^T \mathbf{A}_k^{-1}}{1+z_k\mathbf{e}_k^T \mathbf{A}_k^{-1}\mathbf{e}_k} z_k  \label{Zsca}
\end{equation}

For ease of writing, we introduce the shorthand notation
\begin{align}
& \tilde{\mathbf{A}}_k^{-1} = \frac{ \mathbf{A}_k^{-1} \mathbf{e}_k \mathbf{e}_k^T \mathbf{A}_k^{-1}}{\mathbf{e}_k^T \mathbf{A}_k^{-1}\mathbf{e}_k} \label{TildeA}  \\
& \mathbf{B}_k= \mathbf{Z}_{\rm RL}\left[\mathbf{Z}_{\rm ROT}-\mathbf{Z}_{\rm ROS}\left( \mathbf{A}_k^{-1}-\tilde{\mathbf{A}}_k^{-1}\right) \mathbf{Z}_{\rm SOT}\right] \mathbf{Z}_{\rm TG} \label{B_k}\\
& \mathbf{C}_k=-\mathbf{Z}_{\rm RL} \left[\mathbf{Z}_{\rm ROS}\tilde{\mathbf{A}}_k^{-1} \mathbf{Z}_{\rm SOT} \right] \mathbf{Z}_{\rm TG}\label{C_k} \\
& \mathbf{X}_1\left(z_k\right)=\ \frac{1}{{\sigma^2}}\left(\frac{1}{\chi_k}\mathbf{C}_k \mathbf{Q} \mathbf{B}_k^H +\left(\frac{1}{\chi_k}\mathbf{C}_k \mathbf{Q} \mathbf{B}_k^H\right)^H \right) \label{X1} \\
& \mathbf{X}_2\left(z_k\right) = \frac{1}{\sigma^2 |\chi_k|^2} \mathbf{C}_k \mathbf{Q} \mathbf{C}_k^H \label{X2}
\end{align}
where $\chi_k=\chi_k\left(z_k\right) =1+a_k z_k$ and $a_k=\mathbf{e}_k^T \mathbf{A}_k^{-1} \mathbf{e}_k$. Equations \eqref{TildeA}-\eqref{X2} can be applied if $a_k \ne 0$ and $\chi_k \ne 0$, which is ensured by the physical nature of the problem at hand, and can be tested during the execution of the algorithm.

Hence, the end-to-end channel in \eqref{channel} can be expressed as
\begin{align}
    \mathbf{H}_{\rm E2E} = \mathbf{B}_k+ \mathbf{C}_k/{\chi_k \left(z_k\right)} \label{compx1}
\end{align}
and the achievable data rate in \eqref{rate} can be expressed as
\begin{align}
   & R(z_k) =  \log_2 \left[\det \left(\mathbf{I}_L + \frac{\mathbf{B}_k \mathbf{Q} \mathbf{B}_k^H}{\sigma^2}+\mathbf{X}_1\left(z_k\right) +\mathbf{X}_2\left(z_k\right) \right) \right] \nonumber \\
   & \overset{(a)}{=} \log_2 \left[\det \left( \left[\mathbf{I}_L + \frac{\mathbf{B}_k \mathbf{Q} \mathbf{B}_k^H}{\sigma^2} \right] \mathbf{S}_k\left(z_k\right) \right)\right] \\
   &  = \log_2 \left[\det(\mathbf{I}_L + \frac{\mathbf{B}_k \mathbf{Q} \mathbf{B}_k^H}{\sigma^2})\right]+\log_2 \left[\det(\mathbf{S}_k\left(z_k\right))\right] \label{refeq1}
\end{align}
where $(a)$ is obtained by first applying the eigenvalue decomposition $\mathbf{I}_L+{\mathbf{B}_k \mathbf{Q} \mathbf{B}_k^H}/{\sigma^2}=\mathbf{U}_k \mathbf{\Sigma}_k \mathbf{U}_k^H$, where $\mathbf{U}_k$ is the unitary matrix of eigenvectors and $\mathbf{\Sigma}_k$ is the diagonal matrix of eigenvalues, and $ (\mathbf{U}_k \mathbf{\Sigma}_k \mathbf{U}_k^H)^{-1} = \mathbf{U}_k \mathbf{\Sigma}_k^{-1} \mathbf{U}_k^H$, and by then defining $\mathbf{S}_k\left(z_k\right)=\mathbf{I}_L+ \mathbf{U}_k \mathbf{\Sigma}_k^{-1} \mathbf{U}_k^H \left( \mathbf{X}_1\left(z_k\right) +\mathbf{X}_2\left(z_k\right) \right)$.

In \eqref{refeq1}, only $\mathbf{S}_k=\mathbf{S}_k\left(z_k\right)$ depends on the $k$th tunable impedance $z_k$ to be optimized. Therefore, maximizing the achievable data rate boils down to maximizing $\det(\mathbf{S}_k)$. By applying Sylvester's determinant theorem \cite{sylvester}, we obtain
\begin{align}
    &\hspace{-0.35cm}\det(\mathbf{S}_k\left(z_k\right)) \label{Eq:DetS} \\ &=\det \left(\mathbf{I}_L+ \mathbf{\Sigma}_k^{-\frac{1}{2}} \mathbf{U}_k^H \left(\mathbf{X}_1\left(z_k\right)+\mathbf{X}_2\left(z_k\right)\right)\mathbf{U}_k \mathbf{\Sigma}_k^{-\frac{1}{2}} \right) \nonumber
\end{align}

Denote ${{\bf{D}}_1} =  - {{\bf{Z}}_{\rm RL}}{{\bf{Z}}_{\rm ROS}}{\bf{A}}_k^{ - 1}$, ${{\bf{D}}_2} = {\bf{A}}_k^{ - 1}{{\bf{Z}}_{\rm SOT}}{{\bf{Z}}_{\rm TG}}$. Then, ${{\bf{C}}_k} = {{\bf{D}}_1}{{\bf{e}}_k}{\bf{e}}_k^T{{\bf{D}}_2}/a_k$. By definition of rank, ${\rm{rank}}\left( {{{\bf{C}}_k}} \right) \le \min \left\{ {{\rm{rank}}\left( {{{\bf{D}}_1}} \right),{\rm{rank}}\left( {{{\bf{e}}_k}{\bf{e}}_k^T} \right),{\rm{rank}}\left( {{{\bf{D}}_2}} \right)} \right\}$. Since ${{\rm{rank}}\left( {{{\bf{e}}_k}{\bf{e}}_k^T} \right)}=1$, we obtain ${\rm{rank}}\left( {{{\bf{C}}_k}} \right) = 1$. Without loss of generality, we can then write $\mathbf{C}_k = \mathbf{u}_k\mathbf{v}_k^H$, where
\begin{equation}
    \mathbf{u}_k=-\mathbf{Z}_{\rm RL} \mathbf{Z}_{\rm ROS}\mathbf{A}_k^{-1} \mathbf{e}_k,\; \;  \mathbf{v}_k^H=\mathbf{e}_k^{T}\frac{ \mathbf{A}_k^{-1}}{a_k} \mathbf{Z}_{\rm SOT} \mathbf{Z}_{\rm TG}
\end{equation}

For ease of writing, we introduce the vectors
\begin{equation}
    \tilde{\mathbf{u}}_k =\mathbf{\Sigma}_k^{-\frac{1}{2}} \mathbf{U}_k^H \mathbf{u}_k, \quad \tilde{\mathbf{v}}_k^H = \mathbf{v}_k^H \mathbf{Q} \mathbf{B}_k^H \mathbf{U}_k \mathbf{\Sigma}_k^{-\frac{1}{2}} \label{notation_hat_ukvk}
\end{equation}

By inserting $\mathbf{X}_1\left(z_k\right)$ in \eqref{X1}, $\mathbf{X}_2\left(z_k\right)$ in \eqref{X2} into $\mathbf{S}_k\left(z_k\right)$, and employing the shorthand notation in \eqref{notation_hat_ukvk}, we obtain
\begin{equation}
\mathbf{S}_k\left(z_k\right)=\mathbf{I}_L+\frac{\tilde{\mathbf{u}}_k \tilde{\mathbf{v}}_k^H}{\sigma^2 \chi_k\left(z_k\right)}  + \frac{\tilde{\mathbf{v}}_k \tilde{\mathbf{u}}_k^H}{\sigma^2 \chi_k^*\left(z_k\right)} + \frac{\tilde{\mathbf{u}}_k \tilde{\mathbf{u}}_k^H \mathbf{v}_k^H \mathbf{Q}\mathbf{v}_k}{\sigma^2 |\chi_k\left(z_k\right)|^2}  \label{Sk}
\end{equation}

In \eqref{Sk}, we note that the optimization variable $z_k$ appears only in $\chi_k\left(z_k\right) =1+a_k z_k$, while the vectors and matrices in \eqref{Sk} are independent of $z_k$. The next step is the computation of the determinant of \eqref{Sk}. Since the determinant is invariant to a change of basis functions, $\mathbf{S}_k$ can be expressed in terms of a convenient orthonormal basis that facilitates the computation of $\det(\mathbf{S}_k\left(z_k\right))$. To this end, we apply the Gram-Schmidt orthogonalization \cite[Subsec. 2.6.5]{book_sciComputing} to the vectors $\tilde{\mathbf{u}}_k$ and $\tilde{\mathbf{v}}_k$ in \eqref{notation_hat_ukvk}, since they determine $\mathbf{S}_k\left(z_k\right)$ in \eqref{Sk}. The new set of vectors is denoted by $\mathbf{t}_1$ and $\mathbf{t}_2$, and they are constructed for being orthogonal to each other and to have a unit norm.

In detail, the two vectors $\mathbf{t}_1$ and $\mathbf{t}_2$ are set to $\mathbf{t}_1=\frac{1}{||\tilde{\mathbf{u}}_k||} \tilde{\mathbf{u}}_k$ and $\mathbf{t}_2=\frac{1}{|| \mathbf{t} ||} \mathbf{t}$, with $\mathbf{t} = \tilde{\mathbf{v}}_k- \frac{\langle \tilde{\mathbf{v}}_k , \mathbf{t}_1\rangle}{\langle \mathbf{t}_1 , \mathbf{t}_1\rangle}\mathbf{t}_1 = \frac{||\tilde{\mathbf{u}}_k||^2 \tilde{\mathbf{v}}_k - \tilde{\mathbf{u}}_k^H \tilde{\mathbf{v}}_k\tilde{\mathbf{u}}_k}{||\tilde{\mathbf{u}}_k||^2}$. Accoridngly, $\tilde{\mathbf{u}}_k$ and $\tilde{\mathbf{v}}_k$ in \eqref{notation_hat_ukvk} can be expressed as
\begin{equation}
    \tilde{\mathbf{u}}_k= (\mathbf{t}_1 \; \mathbf{t}_2) \begin{pmatrix} || \tilde{\mathbf{u}}_k || \\ 0  \end{pmatrix}, \quad
    \tilde{\mathbf{v}}_k= (\mathbf{t}_1 \; \mathbf{t}_2) \begin{pmatrix} \mathbf{t}_1^H \tilde{\mathbf{v}}_k \\ \mathbf{t}_2^H \tilde{\mathbf{v}}_k  \end{pmatrix}
\end{equation}

Also, the last three addends in \eqref{Sk} can be reformulated as
\begin{align}
    \tilde{\mathbf{u}}_k \tilde{\mathbf{v}}_k^H
    = (\mathbf{t}_1 \; \mathbf{t}_2) &       \begin{pmatrix}  || \tilde{\mathbf{u}}_k || \tilde{\mathbf{v}}_k^H \mathbf{t}_1 & || \tilde{\mathbf{u}}_k || \tilde{\mathbf{v}}_k^H \mathbf{t}_2\\ 0 & 0   \end{pmatrix}    \begin{pmatrix}  \mathbf{t}_1^H\\ \mathbf{t}_2^H   \end{pmatrix}
\end{align}
\begin{align}
    \tilde{\mathbf{v}}_k \tilde{\mathbf{u}}_k^H =  (\mathbf{t}_1 \; \mathbf{t}_2) &       \begin{pmatrix}  || \tilde{\mathbf{u}}_k || \mathbf{t}_1^H \tilde{\mathbf{v}}_k &  0 \\ ||\tilde{\mathbf{u}}_k || \mathbf{t}_2^H \tilde{\mathbf{v}}_k & 0   \end{pmatrix}    \begin{pmatrix}  \mathbf{t}_1^H\\ \mathbf{t}_2^H   \end{pmatrix}
\end{align}
\begin{align}
    \tilde{\mathbf{u}}_k \tilde{\mathbf{u}}_k^H \mathbf{v}_k^H \mathbf{Q} \mathbf{v}_k
    = (&\mathbf{t}_1 \; \mathbf{t}_2) \begin{pmatrix}  || \tilde{\mathbf{u}}_k ||^2 \mathbf{v}_k^H \mathbf{Q} \mathbf{v}_k  &  0 \\ 0 & 0   \end{pmatrix}   \begin{pmatrix} \mathbf{t}_1^H \\ \mathbf{t}_2^H  \end{pmatrix}
\end{align}

In addition, the identity matrix $\mathbf{I}_L$ can be written as
\begin{equation}
    \mathbf{I}_L= (\mathbf{t}_1 \; \mathbf{t}_2) \begin{pmatrix}  1 &  0 \\ 0 & 1   \end{pmatrix}   \begin{pmatrix} \mathbf{t}_1^H \\ \mathbf{t}_2^H  \end{pmatrix}
\end{equation}

Let $s_k$ be the complex scalar defined as
\begin{equation}
s_k\left(z_k\right)=\frac{|| \tilde{\mathbf{u}}_k || \tilde{\mathbf{v}}_k^H \mathbf{t}_1}{\sigma^2 \chi_k\left(z_k\right)}  + \frac{ || \tilde{\mathbf{u}}_k || \mathbf{t}_1^H \tilde{\mathbf{v}}_k}{\sigma^2 \chi_k^*\left(z_k\right)} + \frac{|| \tilde{\mathbf{u}}_k ||^2 \mathbf{v}_k^H \mathbf{Q} \mathbf{v}_k}{\sigma^2 |\chi_k\left(z_k\right)|^2}
\end{equation}

Accordingly, $\det(\mathbf{S}_k\left(z_k\right))$ in \eqref{Eq:DetS} can be expressed as
\begin{align}
  \det(\mathbf{S}_k){=}&\det \! \left(\!(\mathbf{t}_1 \; \mathbf{t}_2)\! \begin{pmatrix}  1+s_k\left(z_k\right) \; &  \frac{|| \tilde{\mathbf{u}}_k || \tilde{\mathbf{v}}_k^H \mathbf{t}_2}{\sigma^2 \chi_k\left(z_k\right)} \\ \frac{||\tilde{\mathbf{u}}_k || \mathbf{t}_2^H \tilde{\mathbf{v}}_k}{\sigma^2 \chi_k^*\left(z_k\right)} & 1   \end{pmatrix} \!  \begin{pmatrix} \mathbf{t}_1^H \\ \mathbf{t}_2^H  \end{pmatrix}\!\right) \nonumber\\
      \overset{(b)}{=}&\det\begin{pmatrix} 1+ s_k\left(z_k\right) \; &  \frac{|| \tilde{\mathbf{u}}_k || \tilde{\mathbf{v}}_k^H \mathbf{t}_2}{\sigma^2 \chi_k\left(z_k\right)} \\ \frac{||\tilde{\mathbf{u}}_k || \mathbf{t}_2^H \tilde{\mathbf{v}}_k}{\sigma^2 \chi_k^*\left(z_k\right)} & 1   \end{pmatrix} = \det \left(\mathbf{W}\right) \label{Eq:FinalDet}
\end{align}
where $(b)$ follows by defining $\mathbf{T}=(\mathbf{t}_1 \; \mathbf{t}_2)$ and noting that $\det \left( {{\bf{TW}}{{\bf{T}}^H}} \right) = \det \left( {{{\bf{T}}^H}{\bf{TW}}} \right) = \det \left( {\bf{W}} \right)$, since $\det \left( {{{\bf{T}}^H}{\bf{T}}} \right) = {{\bf{I}}_2}$ as $\mathbf{T}$ is a unitary matrix by construction.

Since $\mathbf{W} = \mathbf{W}\left(z_k\right)$ in \eqref{Eq:FinalDet} is a $2 \times 2$ matrix, the determinant of $\mathbf{S}_k=\mathbf{S}_k\left(z_k\right)$ can be expressed in closed-form as
\begin{equation}  \det(\mathbf{S}_k\left(z_k\right))=1+\frac{c_1}{\chi_k\left(z_k\right)} + \frac{c_1^*}{\chi_k^*\left(z_k\right)}  + \frac{c_2}{|\chi_k\left(z_k\right)|^2}
\end{equation}
where
\begin{equation}
    c_1= \frac{|| \tilde{\mathbf{u}}_k || \tilde{\mathbf{v}}_k^H \mathbf{t}_1}{\sigma^2},  \;
    c_2 = \frac{|| \tilde{\mathbf{u}}_k ||^2 \mathbf{v}_k^H \mathbf{Q} \mathbf{v}_k}{\sigma^2} - \frac{|| \tilde{\mathbf{u}}_k ||^2  | \tilde{\mathbf{v}}_k^H \mathbf{t}_2|^2}{\sigma^4} \nonumber
\end{equation}

In conclusion, since $\chi_k\left(z_k\right)=1+a_k z_k$, with $z_k= R_{0,k}+jX_k$ and $R_{0,k}$ is assumed known and fixed, $(\mathbf{P1})$ boils down to maximizing the single-variable (i.e., $X_k$) function
\begin{align}
    f(X_k) & = 1+\frac{c_1}{ 1+a_k(R_{0,k}+jX_k)}  + \frac{ c_1^*}{1+a_k^* (R_{0,k}+jX_k)^*} \nonumber \\ & + \frac{c_2}{|1+a_k (R_{0,k}+jX_k)|^2}, \quad X_k \in \mathcal{P} \label{Eq:f_X}
\end{align}

The optimal solution is stated in the following Proposition.
\begin{proposition} \label{proposition}
Consider the optimization problem
\begin{align}
X_k^{\star}={\arg \max}_{X_k \in \left[X_{\ell b},X_{ub}\right]} f(X_k) \label{OptX}
\end{align}
with $f(X_k)$ in \eqref{Eq:f_X}. Also, let $X_k^1$ and $X_k^2$ be defined as
\begin{align}
X_k^1=\frac{\Im\{\frac{c_1}{a_k^*}\}+2R_{0,k} \Im\{c_1\}+c_2 \Im\{\frac{1}{a_k^*}\}}{2 \left(\Re\{c_1\}+R_{0,k} c_1 a_k^* \right)+c_2} \label{Eq:Xk}
\end{align}
\vspace{-0.20cm}
\begin{align}
X_k^2 &=\frac{\Re\{c_1\} + R_{0,k} \Re\{c_1 a_k^*\}+ \frac{c_2}{2}}{\Re\{c_1\}\Im\{a_k\}-\Re\{a_k\}\Im\{c_1\}} \nonumber \\
&- \frac{\left|c_1 a_k^* \left( \frac{\Re\{a_k\}}{|a_k|^2}+R_{0,k}\right) + \frac{c_2}{2}\right|}{\Re\{c_1\}\Im\{a_k\}-\Re\{a_k\}\Im\{c_1\}}
\end{align}

Define $\bar R(X_k) = R(R_{0,k}+jX_k)$, with $R(\cdot)$ given in \eqref{refeq1}. Then, the optimal solution $X_k^{\star}$ in \eqref{OptX} is the following.
\begin{itemize}
\item If $c_1 a_k^*=c_1^* a_k$ and $\left(2 \Re\{c_1\}+2R_{0,k} c_1 a_k^*+c_2\right) > 0$: \vspace{-0.15cm}
\begin{equation}
X_k^{\star} = \begin{cases} X_k^1 & {\rm{if}} \; X_k^1 \in \left(X_{\ell b},X_{ub}\right) \nonumber \\
 X_{\ell b} & {\rm{if}} \; X_k^1 \le X_{\ell b} \nonumber \\
 X_{ub} & {\rm{if}} \; X_k^1 \ge X_{ub} \nonumber \end{cases} \vspace{-0.15cm}
\end{equation}
\item If $c_1 a_k^*=c_1^* a_k$ and $\left(2 \Re\{c_1\}+2R_{0,k} c_1 a_k^*+c_2\right) < 0$: \vspace{-0.15cm}
\begin{equation}
X_k^{\star} = \begin{cases} X_{\ell b} & {\rm{if}} \; X_k^1 \in \left(X_{\ell b},X_{ub}\right), \; \bar R(X_{\ell b}) \ge \bar R(X_{ub}) \nonumber \\
X_{ub} & {\rm{if}} \; X_k^1 \in \left(X_{\ell b},X_{ub}\right), \; \bar R(X_{\ell b}) < \bar R(X_{ub}) \nonumber \\
X_{ub} & {\rm{if}} \; X_k^1 \le X_{\ell b} \nonumber \\
X_{\ell b} & {\rm{if}} \; X_k^1 \ge X_{ub} \nonumber \end{cases} \vspace{-0.15cm}
\end{equation}
\item If $c_1 a_k^*=c_1^* a_k$ and $\left(2 \Re\{c_1\}+2R_{0,k} c_1 a_k^*+c_2\right) = 0$: \vspace{-0.15cm}
\begin{equation}
X_k^{\star} = \begin{cases} X_{\ell b} & {\rm{if}} \; \bar R(X_{\ell b}) \ge \bar R(X_{ub}) \nonumber \\
X_{ub} & {\rm{if}} \; \bar R(X_{\ell b}) < \bar R(X_{ub}) \nonumber
\end{cases} \vspace{-0.15cm}
\end{equation}
\item If $c_1 a_k^* \neq c_1^* a_k$ and $\Re\{c_1\} \Im\{a_k\} > \Re\{a_k\} \Im\{c_1\}$: \vspace{-0.15cm}
\begin{equation}
X_k^{\star}=
\begin{cases}
X_k^2 & {\rm{if}} \; X_k^2 \in \left(X_{\ell b},X_{ub}\right), \; \bar R(X_k^2) \ge \bar R(X_{ub}) \nonumber \\
X_{ub} & {\rm{if}} \; X_k^2 \in \left(X_{\ell b},X_{ub}\right), \; \bar R(X_k^2) < \bar R(X_{ub}) \nonumber \\
X_{\ell b} & {\rm{if}} \; X_k^2 \le X_{\ell b}, \; \bar R(X_{\ell b}) \ge \bar R(X_{ub}) \nonumber \\
X_{ub} & {\rm{if}} \; X_k^2 \le X_{\ell b}, \; \bar R(X_{\ell b}) < \bar R(X_{ub}) \nonumber \\
X_{ub} & {\rm{if}} \; X_k^2 \ge X_{ub}
\nonumber \end{cases} \vspace{-0.15cm}
\end{equation}
\item If $c_1 a_k^* \neq c_1^* a_k$ and $\Re\{c_1\} \Im\{a_k\} < \Re\{a_k\} \Im\{c_1\}$: \vspace{-0.15cm}
\begin{equation}
X_k^{\star} =
\begin{cases}
X_k^2 & {\rm{if}} \; X_k^2 \in \left(X_{\ell b},X_{ub}\right), \; \bar R(X_k^2) \ge \bar R(X_{\ell b}) \nonumber \\
X_{\ell b} & {\rm{if}} \; X_k^2 \in \left(X_{\ell b},X_{ub}\right), \ \bar R(X_k^2) < \bar R(X_{\ell b}) \nonumber \\
X_{ub} & {\rm{if}} \; X_k^2 \ge X_{ub}, \; \bar R(X_{ub}) \ge \bar R(X_{\ell b}) \nonumber
\\
X_{\ell b} & {\rm{if}} \; X_k^2 \ge X_{ub}, \; \bar R(X_{ub}) < \bar R(X_{\ell b}) \nonumber
\\
X_{\ell b} & {\rm{if}} \; X_k^2 \le X_{\ell b}
\nonumber \end{cases}
\end{equation}
\end{itemize}
\end{proposition}
\begin{proof}
It follows by computing the first- and second-order derivatives of $f(X_k)$ as a function of  $X_k$, and by analyzing when the stationary points are maxima in the feasible set.
\end{proof} \vspace{-0.25cm}

Based on Proposition \ref{proposition}, the proposed complete algorithm for iteratively solving $\mathbf{(P0)}$ is given in Algorithm \ref{algo}.

\begin{algorithm}[!t] \footnotesize
\caption{Proposed algorithm for solving $\mathbf{(P0)}$}\label{algo}
\begin{algorithmic}
\State  {\bf{Input}}: Compute the impedance matrices from \cite[Lemma~2]{marco_EuCAP};
\State {\bf{Initialize}}: $q=0$, $\epsilon \geq 0$, $\mathbf{r}_0=[R_{0,1}, \hdots, R_{0,N_{\rm RIS}}]^T$, $\mathbf{x}^{(0)}=[X_1^{(0)}, \hdots, X_{N_{\rm RIS}}^{(0)}]^T \in \mathcal{P}^{N_{\rm RIS}}$, $R^{(-1)}=0$, $R^{(0)}=R(\mathbf{Q}^{(0)}, \mathbf{Z}^{(0)})$ with $\mathbf{Z}^{(0)} = \diag(\mathbf{r}_0) + j \diag(\mathbf{x}^{(0)})$ and $R(\cdot,\cdot)$ defined in \eqref{rate};
\While {$|R^{(q)}-R^{(q-1)}| > \epsilon$}
\State Compute $\mathbf{Q^{\star}}$ from \eqref{waterfilling};
\For{$k=1, \hdots, N_{\rm RIS}$}
\State Compute $X_k^{\star}$ from Proposition \ref{proposition};
\State Update $\mathbf{Z}_{\rm RIS}^{\star}(k,k) \leftarrow R_{0,k} + j X_k^{\star}$;
\EndFor
\State $q=q+1$, $R^{(q)}=R(\mathbf{Q^{\star}},\mathbf{Z}_{\rm RIS}^{\star})$;
 \EndWhile
\State  {\bf{Return}}: $\mathbf{Q^{\star}}$ and $\mathbf{Z}_{\rm RIS}^{\star}$.
\end{algorithmic}
\end{algorithm}
\setlength{\textfloatsep}{3pt}%

\vspace{-0.25cm}
\section{Complexity and Convergence} \vspace{-0.15cm}
\subsection{Computational Complexity} \vspace{-0.1cm}
We evaluate the computational complexity per iteration (i.e., for one iteration of the {\bf{while}} loop in Algorithm \ref{algo}) in terms of complex multiplications. The complexity is determined by the number of multiplications needed to compute $\mathbf{Q^{\star}}$ and by $N_{\rm{RIS}}$ times (because of the {\bf{for}} loop) the number of multiplications needed to compute $X_k^{\star}$. For simplicity, we assume $L \le M$, $L \ll N_{\rm RIS}$, $M \ll N_{\rm RIS}$. The complexity of $\mathbf{Q^{\star}}$ is determined by the computation of $\mathbf{H}_{\rm E2E}$ in \eqref{channel}, whose complexity is determined by the product $\mathbf{Z}_{\rm ROS} \mathbf{Z}_{\rm sca} \mathbf{Z}_{\rm SOT}$. Thus, the complexity is $\mathcal{O}(N_{\rm RIS}^3+N_{\rm RIS}^2L)$. The complexity of $X_k^{\star}$ is determined by the computations of $\mathbf{A}_k^{-1}$, $\mathbf{B}_k$, $\mathbf{C}_k$, as per \eqref{TildeA}-\eqref{X2}, whose complexities are $\mathcal{O}(N_{\rm RIS}^3)$, $\mathcal{O}(N_{\rm RIS}^2L)$, $\mathcal{O}(N_{\rm RIS}^2L)$. Thus, the complexity is $\mathcal{O}(N_{\rm RIS}^3+N_{\rm RIS}^2L)$. In conclusion, the complexity of Algorithm \ref{algo} is $\mathcal{O}\left((N_{\rm RIS}+1)(N_{\rm RIS}^3+N_{\rm RIS}^2 L)\right)$.

With similar approximations, the complexities of the algorithms in \cite{abrardo_MIMO}  and \cite{placido_MISO} are $\mathcal{O}(2N_{\rm RIS}^3+N_{\rm RIS}^2 L^3)$ and $\mathcal{O}\left(2N_{\rm RIS}^3+N_{\rm RIS}^2(LM+L+M)\right)$, respectively. Since Algorithm \ref{algo} optimizes the RIS elements one by one iteratively, the complexity scales as $\mathcal{O}\left(N_{\rm RIS}^4\right)$. The complexities of \cite{abrardo_MIMO}  and \cite{placido_MISO} scale as $\mathcal{O}\left(N_{\rm RIS}^3\right)$, since all the RIS elements are optimized at once. The overall complexity depends, however, on the number of iterations to converge and the amount of time (in seconds) that each algorithm needs per iteration. In contrast to \cite{abrardo_MIMO}  and \cite{placido_MISO}, $X_k^{\star}$ is available in closed-form. In Sec. V, we show that Algorithm \ref{algo} needs few iterations to converge.

\vspace{-0.35cm}
\subsection{Convergence}
Algorithm \ref{algo} exploits the BCD method. Specifically, the objective function and the constraints in $(\mathbf{P1})$ are continuous and differentiable, the constraints have a decomposable (decoupled) structure in the optimization variables, and the feasible set of each optimization variable is closed and convex. Also, according to Proportion \ref{proposition}, the solution $X_k^{\star}$ in \eqref{OptX} is the unique optimum. Based on \cite[Prop. 2.7.1]{convergence1}, therefore, Algorithm \ref{algo} converges to a stationary point of $(\mathbf{P1})$ and $(\mathbf{P0})$.

\vspace{-0.25cm}
\section{Numerical Results}
{\textcolor{black}{The simulation setup and parameters are the same as in \cite{placido_MISO} to facilitate comparison}. Specifically, we consider a 4-antenna transmitter whose center is located at $(0,0) \lambda$, a single-antenna receiver located at $(9.6, 14.4) \lambda$, and an RIS whose center is located at $(0,24) \lambda$ where $\lambda=10$~cm. The inter-distance at the transmitter is $\lambda/2$. The transmit and receive antennas, RIS elements, and scattering objects are identical thin wires of length $l=\lambda/2$ and radius $a=\lambda/500$. All of them are oriented as in \cite[Fig. 1]{placido_MISO}. We set $R_{0,k}=0.2$~Ohm, $\forall k$, $P_t=21$~dBm, $\sigma^2=-80$~dBm, $\mathcal{P}=[-302.50, -19.66]$~Ohm, $\mathbf{Z}_G=\mathbf{Z}_L= 50\mathbf{I}_M$~Ohm, $\mathbf{Z}_{\rm US}=\mathbf{0}_{N_e}$~Ohm. Also, we consider the presence of $4$ randomly distributed clusters each containing $N_e=50$ scattering objects. The direct link is ignored due to the presence of obstacles. The results are averaged over 100 independent realizations for the locations of the scattering objects. In Algorithm \ref{algo}, $\mathbf{x}^{(0)}$ is initialized at random in the feasible set and $\mathbf{Q}^{(0)}$ is computed from \eqref{waterfilling} given $\mathbf{x}^{(0)}$.

{\textcolor{black}{In Fig.~\ref{fig:convergence1}, we illustrate the rate when the inter-distance and the number of RIS elements are configured for ensuring that the size of the RIS is the same. Specifically, we consider the case study when the mutual coupling is taken into account at the design stage (MCA)}}. We see that Algorithm \ref{algo} provides superior performance compared with the algorithms in \cite{abrardo_MIMO} and \cite{placido_MISO}. Specifically, Algorithm \ref{algo} (i) converges faster and (ii) reaches a higher value of rate. Both benefits are attributed to two features of Algorithm \ref{algo}: avoiding the Neumann series approximation and using closed-form expressions for $\mathbf{Q^{\star}}$ and $\mathbf{Z}_{\rm RIS}^{\star}$ at each iteration. Similar to \cite{abrardo_MIMO} and \cite{placido_MISO}, we see the benefits of reducing the inter-distance of the RIS elements and considering the mutual coupling at the optimization stage.

To better evaluate the execution time of Algorithm \ref{algo} and compare it against \cite{abrardo_MIMO} and \cite{placido_MISO}, Table~\ref{computation_time} shows the time (in seconds) that the algorithms need to converge. Specifically, the algorithms are deemed to have converged if the increment of the rate in two consecutive iterations is less than $10^{-4}$. First, we note that the three algorithms never cross each other, even at convergence: Algorithm \ref{algo} reaches always the highest value of rate, whereas, at convergence, the algorithms in \cite{abrardo_MIMO} and \cite{placido_MISO} reach (on average with respect to $d$) the $90\%$ and $98\%$ of the rate provided by Algorithm \ref{algo}. In Table~\ref{computation_time}, we report the amount of time that the algorithms in \cite{abrardo_MIMO} and \cite{placido_MISO} need to reach convergence, and the amount of time that is required for Algorithm \ref{algo} to reach the $90\%$ and $98\%$ of the rate that it achieves at convergence. We see the superiority of Algorithm \ref{algo}, especially for small values of the inter-distance $d$.

{\textcolor{black}{In Fig.~\ref{fig:convergence2}, we show the rate when $N_{\rm{RIS}}$ is kept fixed, and the size of the RIS decreases when $d$ decreases. We report only the results for Algorithm \ref{algo}, since the algorithms in \cite{abrardo_MIMO} and \cite{placido_MISO} provide similar trends as those shown in Fig.~\ref{fig:convergence1}. We compare the case study when the mutual coupling is disregarded (MCU) against the case study MCA, similar to \cite{xuewen_SISO}. In the MCU case, the rate decreases as the mutual coupling becomes more significant ($d$ decreases). In the MCA case, the rate has a non-monotonic behavior with $d$. Notably, the setup with $d = \lambda/2$ (negligible mutual coupling) provides almost the same rate as the setup with $d = \lambda/16$. Thus, the size of the RIS can be reduced by a factor of eight while keeping $N_{\rm{RIS}}$ fixed}}.

\begin{figure}[!t]
\center
\includegraphics[width=0.73\columnwidth]{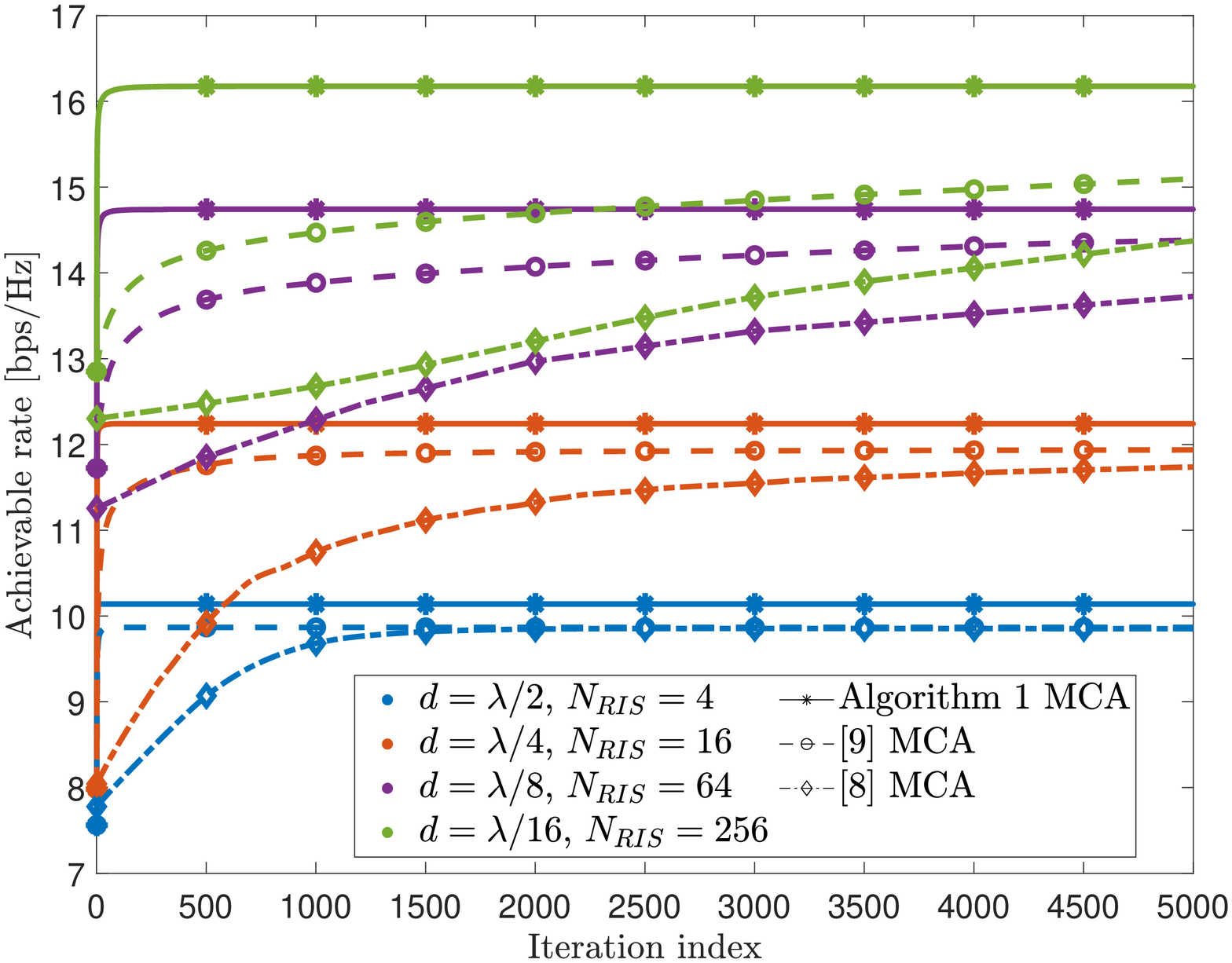} \vspace{-0.25cm}
\caption{\footnotesize Fig. 1: Convergence of the algorithms (the size of the RIS is equal)}
\label{fig:convergence1} \vspace{-0.25cm}
\end{figure}
\begin{table}[!t]
\centering
\caption{\footnotesize Table II: Comparison of the execution time [seconds]} \vspace{-0.2cm} \label{computation_time}\scalebox{0.9}{
\begin{tabular}{|>{\columncolor[gray]{1.0}}c|c|c|}
\hline
\bf{d} &  \bf{Algo.~\ref{algo} ($90\%$)} & \bf{\cite{abrardo_MIMO}} \\
\hline \hline
$\lambda/2$ & 0.001 & 0.800 \\
\hline
$\lambda/4$ & 0.004 & 0.770  \\
\hline
$\lambda/8$ & 0.167
& 8.135 \\
\hline
 $\lambda/16$ & 16.530 & 213.128 \\
 \hline
\end{tabular}
}
\hspace{0.01in}
\centering \scalebox{0.9}{
\begin{tabular}{|>{\columncolor[gray]{1.0}}c|c|c|}
\hline
\bf{d} & \bf{Algo.~\ref{algo} ($98\%$)} & \bf{\cite{placido_MISO}} \\
\hline \hline
$\lambda/2$ & 0.001 & 0.0154 \\
\hline
$\lambda/4$ & 0.008 & 0.896  \\
\hline
$\lambda/8$ & 0.834
& 27.686 \\
\hline
 $\lambda/16$ & 170.807 & 946.404 \\
 \hline
\end{tabular}
}
\end{table}
\begin{figure}[!t]
\center
\includegraphics[width=0.73\columnwidth]{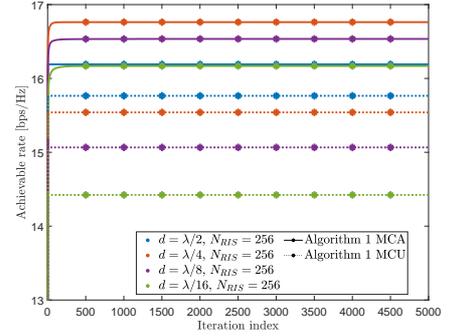} \vspace{-0.25cm}
\caption{\footnotesize Fig. 2: Convergence of the algorithms  (the number of RIS elements is equal)}
\label{fig:convergence2}
\end{figure}

\vspace{-0.25cm}
\section{Conclusion} \vspace{-0.15cm}
{\textcolor{black}{By assuming a discrete thin wire dipole model, we have proposed a novel and provably convergent algorithm for optimizing RIS-assisted MIMO systems in the presence of scattering objects in the environment and mutual coupling among the RIS elements. As for the impact of mutual coupling, three main conclusions can be drawn: (1) if the size of the RIS is kept fixed and the number of RIS elements is increased, the rate increases if the RIS is optimized by taking the mutual coupling into account; (2) if the number of RIS elements is kept fixed and the inter-distance is reduced, the physical size of the RIS can be reduced with no performance degradation with respect to the typical $d = \lambda/2$ configuration, provided that the RIS is optimized by taking the mutual coupling into account; (3) if the mutual coupling is ignored, the rate decreases. Thus, it is important to model the mutual coupling accurately and to duly take it into account when optimizing the RIS}}.

\vspace{-0.25cm}
\bibliographystyle{IEEEtran}
\bibliography{biblio_new}

\end{document}